\documentclass[letterpaper, 10 pt, conference,twocolumn]{ieeeconf}
\IEEEoverridecommandlockouts
\usepackage{amsmath,amsfonts,amssymb}
\usepackage{graphicx,epsfig,psfrag,subfigure}
\usepackage{xfrac,bm}
\usepackage{algorithmic,algorithm}
\usepackage{subfigure}
\usepackage{comment,cite}
\usepackage{xfrac}
\usepackage{xcolor}
\usepackage{mathrsfs}
\IEEEoverridecommandlockouts 

\DeclareMathOperator*{\argmin}{arg\,min}

\DeclareMathOperator{\proj}{\mathbf{Proj}}

\newcommand{\transpose}{{\top}}

\newcommand{\x}{{\bm{x}}}
\newcommand{\y}{{\bm{y}}}
\newcommand{\z}{{\bm{z}}}

\newcommand{\uu}{{\bm{u}}}
\newcommand{\mbx}{\mathbf{x}}
\newcommand{\mbz}{\mathbf{z}}

\newcommand{\cA}{\mathcal{A}}
\newcommand{\cB}{\mathcal{B}}

\newcommand{\cD}{\mathcal{D}}

\newcommand{\cH}{\mathcal{H}}

\newcommand{\cP}{\mathcal{P}}

\newcommand{\cR}{\mathcal{R}}

\newcommand{\cT}{\mathcal{T}}
\newcommand{\cU}{\mathcal{U}}

\newcommand{\bR}{\mathbb{R}}

\newcommand{\dist}{{\mathrm{dist}}}

\DeclareSymbolFont{symbolsC}{U}{pxsyc}{m}{n}

\newtheorem{remark}{Remark}
\newtheorem{proposition}{Proposition}

\newtheorem{definition}{Definition}

\newtheorem{fact}{Fact}

\def\BibTeX{{\rm B\kern-.05em{\sc i\kern-.025em b}\kern-.08em
    T\kern-.1667em\lower.7ex\hbox{E}\kern-.125emX}}

\title{\textbf{Optimal Strategies for Guarding a Compact and Convex Target Set:\\ A Differential Game Approach}}

\author{Yoonjae Lee \and Efstathios Bakolas \thanks{Y. Lee (PhD student) and E. Bakolas (Associate Professor) are with the Department of Aerospace Engineering
and Engineering Mechanics, The University of Texas at Austin,
Austin, Texas 78712-1221, USA, Emails: yol033@utexas.edu; bakolas@austin.utexas.edu}}

\begin{document}

\maketitle

\begin{abstract}
We revisit the two-player planar target-defense game initially posed by Isaacs where a pursuer (or defender) attempts to guard a target set from an attack by an evader (or attacker). This paper builds on existing analytical solutions to games of defending a simple shape of target area to develop a generalized and extended solution to the same game with a compact convex target set with smooth boundary. Isaacs' method is applied to address the game of kind and games of degree. A geometric solution approach is used to find the barrier surface that demarcates the winning sets of the players. A value function coupled with a set of optimal state feedback strategies in each winning set is derived and proven to correspond to the saddle point solution of the game. The proposed solutions are illustrated by means of numerical simulations.
\end{abstract}

\section{Introduction}

Problems of defending a target set from attacks by enemies have received significant attention due to their potential applications in aerospace and robotics. Problems of this type, which involve adversarial agents, are typically formulated as differential games, as was firstly done by Isaacs in \cite{isaacs1965differential}. Therein, Isaacs poses a target-defense game between a pursuer (or defender) and an evader (or attacker) which subsequently decomposes into two subproblems: a game of kind and a game of degree. The former game deals with the question of controllability (e.g., can the capture/attack ever occur?), whereas the latter game deals with the computation of the optimal strategies of the players (e.g., what is the optimal action of the pursuer/evader to successfully defend/attack the target set?). If the players have simple motion, one can solve the game of kind by using geometric methods, such as Voronoi diagrams and Apollonius circles \cite{bakolas2010optimal,bakolas2012relay,makkapati2018optimal,makkapati2019optimal}. For the game of degree, the value of the game as well as the saddle point strategies of the players must be derived from (or shown to satisfy) the Hamilton-Jacobi-Isaacs (HJI) partial differential equation.

There have been extensive work on target-defense games (also referred to as reach-avoid games) combining geometric analysis and Isaacs' method. In \cite{yan2017defense}, Yan et al. identify a barrier curve that demarcates the winning regions of the defender and attacker in a game of defending a circular perimeter. This work is revisited and enhanced by Garcia et al. using Isaacs' method in \cite{garcia2019optimal,garcia2020pride}. In \cite{yan2018reach,yan2019task}, Yan et al. extend their method of constructing barriers to a multiplayer border-defense game that takes place in a compact convex game region. In \cite{garcia2019cooperative,von2020multiple,garcia2020multiple}, the authors apply Isaacs' method to multiplayer border-defense games and address the presence of dispersal surfaces wherein the players may have not have an unique optimal strategy. A multiplayer reach-avoid game under information asymmetry is studied in \cite{selvakumar2019feedback}.

All of the aforementioned references, however, only discuss the particular cases where the shape of the target set is simple (e.g., a point, line, or circle). In \cite{shishika2018local} and \cite{shishika2020cooperative}, Shishika et al. study multiplayer perimeter-defense games with any shape of compact convex target set, but the defenders' motion is constrained along the target boundary. Similar problems with motion constraints have also been studied in \cite{shishika2019team,von2020guarding,lee2020perimeter}. In \cite{pachter2017differential}, the authors study games of defending a compact convex target set, providing a geometric solution approach to the game of kind for the same speed player case. However, a barrier surface is only identified for a circular target case, and the game of degree is only solved for a singleton case.

The main contribution of this paper is to provide a complete solution to both the game of kind and the game(s) of degree of the planar Two-Player Target-Defense Game (TPTDG) with any shape of compact convex target set with smooth boundary. In particular, we generalize and extend the solution framework and results in \cite{garcia2019optimal} specifically proposed for the circular target case. The barrier function we propose takes a generalized form of (11) in \cite{garcia2019optimal} such that it can be applied to any target-defense game (as long as the target set is compact and convex and has a smooth boundary). Furthermore, our proposed state feedback strategies for the pursuer and evader are also in the generalized form of their counterparts in \cite{garcia2019optimal} but are adapted to the problem formulation of this paper. These generalized feedback control laws are nevertheless shown to correspond to the saddle point solution of the game.

The rest of the paper is structured as follows. In Section~\ref{sec:problem_formulation}, a two-player target-defense game with a compact convex target set with smooth boundary is formulated. In Section~\ref{sec:game_of_kind}, a geometric solution approach is introduced and used to solve the game of kind. In Section~\ref{sec:games_of_degree}, the optimal strategies of the two players in the capture and attack games are proposed and proven to be the saddle point solutions. Section~\ref{sec:simulations} presents numerical simulations are finally, Section~\ref{sec:concl} summarizes the results of the paper and introduces directions for follow-up research.

\section{Problem Formulation} \label{sec:problem_formulation}

A planar Two-Player Target-Defense Game (TPTDG) with a compact convex target set with smooth boundary is considered. The domain of the game region is taken to be the Euclidean 2-D plane, $\bR^2$, whereas the target set is a compact convex set $\Omega \subset \bR^2$, which is defined as
\begin{equation}
    \Omega := \{ \z = (x,y) \in \bR^2 : h(\z) \leq 0 \},
\end{equation}
where $h : \bR^2 \rightarrow \bR$ is a convex and smooth function. The boundary of $\Omega$, along which $h(\z) = 0$ for $\z \in \bR^2$, is denoted by $\partial \Omega$. The goal of the evader is to enter $\Omega$ (or reach $\partial \Omega$), whereas the pursuer's goal is to capture the evader before the evader enters the set. Both players are assumed to have complete information about each other's state and dynamics at every time instance. The set $\Omega$ is assumed to be permeable, which means that the pursuer can freely move in and out of the set if it is on his way to capture the evader. The players are assumed to have simple motion, that is,
\begin{align}
    \dot{\x}_P  &= v_P \uu_P,  &\x_P(0) & =\x_{P,0}, \label{eq:motion1}
    \\
    \dot{\x}_E  &= v_E \uu_E,  &\x_E(0) & =\x_{E,0}, \label{eq:motion2}
\end{align}
where $\x_P = (x_P,y_P) \in \bR^2$ and $\x_E = (x_E,y_E) \in \bR^2$ (resp., $\x_{P,0} = (x_{P,0},y_{P,0}) \in \bR^2$ and $\x_{E,0} = (x_{E,0},y_{E,0}) \in \bR^2 \backslash \Omega$) denote the positions of the pursuer and the evader at time $t$ (resp., at time $t=0$), respectively. Furthermore, $v_P \in \bR_{\geq 0}$ (resp., $v_E \in \bR_{\geq 0}$) denotes the maximum allowable speed of the pursuer (resp., the evader) and $\uu_P = (u_{P,x},u_{P,y}) \in \cU$ and $\uu_E = (u_{E,x},u_{E,y}) \in \cU$ denote the control input of the pursuer and the evader, respectively, where $\cU = \{ \uu \in \bR^2 : \| \uu \| = 1 \}$ is the common input set ($\|\cdot\|$ denotes the standard Euclidean vector norm). Let us define the speed ratio $\gamma := v_E/v_P$, then only the slower evader case (i.e., $0 < \gamma < 1$)
is considered in this paper. Let us denote the joint state of the players (or game state) by $\mbx = (\x_P, \x_E) \in \mathbb{R}^4$, then the corresponding game dynamics can be written as
\begin{equation}\label{eq:statedyns}
\dot{\mbx} = \bm{f}(\mbx,\uu_P,\uu_E),~~~~\mbx(0)=\mbx_0,
\end{equation}
where $\bm{f}: \bR^4 \times \cU \times \cU \rightarrow \bR^4$ is the vector field of the game dynamics, i.e., $\bm{f}(\mbx,\uu_P,\uu_E) = (v_P \uu_P, v_E \uu_E)$, and $\mbx_0 = (\x_{P,0}, \x_{E,0})$ is the initial state of the game.

The TPTDG has two terminal conditions: one for capture and another one for attack. Capture occurs as the game state $\mbx$ reaches the capture manifold $\cT_c := \{ \mbx = (\x_P, \x_E) \in \bR^4: (\x_P = \x_E) \wedge (\x_E \notin \Omega) \}$, whereas attack occurs as the same state reaches the attack manifold $\cT_a := \{ \mbx = (\x_P, \x_E) \in \bR^4 : \x_E \in \Omega \}$. Let $t_f := \min(t_c,t_a)$, where $t_c = \inf \{ t \geq 0 : \mbx(t) \in \cT_c \}$ and $t_a = \inf \{ t \geq 0 : \mbx(t) \in \cT_a \}$, then we say that the pursuer (resp., the evader) has won the TPTDG if $\mbx(t_f) \in \cT_c$ (resp., if $\mbx(t_f) \in \cT_a$). As will be discussed in the following section, if both players adhere to their optimal strategies (referred to as optimal play), the result of the TPTDG can be predicted. If capture is expected (referred to as capture game), the objective of the pursuer (resp., the evader) is to maximize (resp., minimize) the minimum distance between the point of capture and the set $\Omega$. The payoff functional corresponding to capture game is given by
\begin{equation}\label{eq:payofffun1}
    J_c \left( \mbx_0,\uu_P(\cdot),\uu_E(\cdot) \right) = \dist \left( \x_E(t_c),\Omega \right),
\end{equation}
where $\dist: \mathbb{R}^2 \times \cP(\Omega)  \rightarrow \bR_{\geq 0}$, $\mathrm{dist}(\x, S)= \inf_{y \in S}\|\x - \y\|$
is the distance function that measures the closeness of a point $\x \in \bR^2$ from a set $S \subseteq \Omega$, or $S \in \cP(\Omega)$ (the symbol $\cP(\Omega)$ denotes the powerset of $\Omega$). Since the set $\Omega$ is assumed to be compact and convex, there always exists a unique point $\z\in S$ such that $\mathrm{dist}(\x, S)= \inf_{\y \in S}\|\x - \y\| = \min_{\y \in S}\|\x - \y\| = \|\x- \z\|$, which implies that $\dist(\x, S)$ corresponds to the minimum distance between the given point $\x$ and the compact convex set $S$. The saddle point of this game is a pair of state feedback strategies $\uu_P^\star(\cdot)$ and $\uu_E^\star(\cdot)$ that satisfy
\begin{equation}\label{eq:payofffun1minmax}
    \min_{\uu_E^\star(\cdot)}\max_{\uu_P(\cdot)}J_c \leq \min_{\uu_E^\star(\cdot)}\max_{\uu_P^\star(\cdot)}J_c \leq \min_{\uu_{E}(\cdot)}\max_{\uu_P^\star(\cdot)}J_c,
\end{equation}
and the value funciton of the game is given by
\begin{equation} \label{eq:valfun1}
    V_{c}(\mbx_0) = \min_{\uu_E(\cdot)} \max_{\uu_P(\cdot)} J_c.
\end{equation}
When attack is expected (referred to as attack game), the objective of the pursuer (resp., the evader) is to minimize (resp., maximize) the distance between the pursuer's final position and the point of attack. For the attack game, the payoff functional is given by
\begin{equation} \label{eq:payofffun2}
    J_a(\mbx_0,\uu_P(\cdot),\uu_E(\cdot) ) = \| \x_E(t_a) - \x_P(t_a) \|,
\end{equation}
the saddle point satisfies
\begin{equation} \label{eq:payofffun2minmax}
    \min_{\uu_P^\star(\cdot)}\max_{\uu_E(\cdot)}J_a \leq \min_{\uu_P^\star(\cdot)}\max_{\uu_E^\star(\cdot)}J_a \leq \min_{\uu_P(\cdot)}\max_{\uu_E^\star(\cdot)}J_a,
\end{equation}
and the value function is given by
\begin{equation}\label{eq:valfun2}
    V_a(\mbx_0) = \min_{\uu_P(\cdot)} \max_{\uu_E(\cdot)} J_a.
\end{equation}

\section{The Game of Kind} \label{sec:game_of_kind}

In this section, we address the game of kind, that is, we determine the conditions under which the pursuer or the evader can win the TPTDG. To this end, the barrier surface that divides the state space of the game $\bR^4$ into two sets, $\cR_c$ (capture set) and $\cR_a$ (attack set), is identified such that the pursuer is guaranteed to win the game in $\cR_c$ (i.e., there exists an optimal strategy to successfully capture the evader), whereas the evader is guaranteed to win the game in $\cR_a$ (i.e., there exists an optimal strategy to successfully attack the target set). In \cite{isaacs1965differential}, Isaacs suggests a geometric method for the slower evader case, namely the Apollonius circle, which can be used to divide the game region $\bR^2$ into the dominant regions of the pursuer and the evader, respectively, given the current positions of the two players. This geometric solution approach is justified via Hamiltonian analysis as below.
\begin{proposition}\label{prop:hamiltonian}
    For players with kinematics defined in \eqref{eq:motion1} and \eqref{eq:motion2}, respectively, and payoff functional given by either \eqref{eq:payofffun1} or \eqref{eq:payofffun2}, the players' optimal control inputs are constant over time and optimal trajectories are straight lines.
\end{proposition}
\begin{proof}
    The following analysis is based on the capture game (the analysis for the attack game can be done mutatis mutandis). Let us consider the Hamiltonian function $\cH: \bR^2 \times \bR^2 \times \mathcal{U} \times \mathcal{U}$ with
    \begin{align}
        \cH(\bm{\lambda}_P, \bm{\lambda}_E, \uu_P, \uu_E) = v_P \bm{\lambda}_P^\transpose \uu_P + v_E \bm{\lambda}_E^\transpose \uu_E,
    \end{align}
    where $\bm{\lambda}_P \in \bR^2$ and $\bm{\lambda}_E \in \bR^2$ are the co-state vectors. Let $\uu_P^*(t)$ and $\uu_E^*(t)$ be the optimal inputs (in open-loop form), then the corresponding optimal costates satisfy the canonical equations: $\dot{\bm{\lambda}}_P^*(t) \equiv 0$ and $\dot{\bm{\lambda}}_E^*(t) \equiv 0$, where the fact that the Hamiltonian does not depend on the states has been used. This implies $\bm{\lambda}_P^*(t) \equiv \bar{\bm{\lambda}}_P$ and $\bm{\lambda}_E^*(t) \equiv \bar{\bm{\lambda}}_E$, where $\bar{\bm{\lambda}}_P$ and $\bar{\bm{\lambda}}_E$ are constant non-zero vectors in $\bR^2$. Additionally,
    \begin{align}\label{eq:PMP}
        \cH(\bar{\bm{\lambda}}_P, \bar{\bm{\lambda}}_E&, \uu_P^*(t), \uu_E^*(t))  = \min_{\nu_E \in \cU}\max_{\nu_P\in\cU} \cH(\bar{\bm{\lambda}}_P, \bar{\bm{\lambda}}_E, \nu_P, \nu_E) \nonumber \\
        &\qquad\qquad\qquad~ = \max_{\nu_P \in \cU} \min_{\nu_E\in \cU} \cH(\bar{\bm{\lambda}}_P, \bar{\bm{\lambda}}_E, \nu_P, \nu_E) \nonumber \\
        &\qquad~~ = \max_{\nu_P \in \cU} v_P \bar{\bm{\lambda}}_P^\transpose \nu_P + \min_{\nu_E\in \cU} v_E \bar{\bm{\lambda}}_E^\transpose \nu_E,
    \end{align}
    for all $t \in [0,t_c]$. In view of the Cauchy-Schwartz inequality, it follows readily from \eqref{eq:PMP} that 
    \begin{equation}
    \uu_P^*(t) \equiv \dfrac{\bar{\bm{\lambda}}_P}{\|\bar{\bm{\lambda}}_P\|},~~~~~
    \uu_E^*(t)\equiv -\dfrac{\bar{\bm{\lambda}}_E}{\|\bar{\bm{\lambda}}_E\|}.
    \end{equation}
   This completes the proof.
\end{proof}

\begin{remark}
    Proposition~\ref{prop:hamiltonian} implies that, under optimal play, the ratio of the travel distance of the pursuer to that of the evader is constant, i.e., $\| \x_E(t) - \x_{E,0} \| = \gamma \| \x_P(t) - \x_{P,0} \|$, for all $t \in [0,t_f]$.
\end{remark}

Since the slower evader case is considered herein, the Apollonius circle will be used frequently in the subsequent analysis, whose definition is provided next.

\begin{definition}[Apollonius Circle] \label{def:apolloncirc}
    Given the TPTDG defined in Section~\ref{sec:problem_formulation}, the Apollonius circle between the two positions $\x_P = (x_P,y_P) \in \bR^2$ and $\x_E = (x_E,y_E) \in \bR^2$ is the set of points at which both players will simultaneously arrive under optimaly play (i.e., with constant control inputs), that is,
    \begin{align} \label{eq:apolcircle}
        &\cA(\mbx;\gamma) := \left\{ \z \in \bR^2 : \| \z - \bm{C}(\mbx;\gamma) \| = R(\mbx;\gamma) \right\}.
    \end{align}
    The center, $\bm{C}(\mbx;\gamma)$, and the radius, $R(\mbx;\gamma)$, of the Apollonius circle are respectively given by
    \begin{align}
        &\bm{C}(\mbx;\gamma) = \left( \frac{1}{1-\gamma^2} (x_E-\gamma^2 x_P), \frac{1}{1-\gamma^2} (y_E-\gamma^2 y_P) \right), \label{eq:apolcircle_center}
        \\
        &R(\mbx;\gamma) = \frac{\gamma}{1-\gamma^2} \sqrt{(x_E-x_P)^2 + (y_E-y_P)^2}. \label{eq:apolcircle_radius}
    \end{align}
\end{definition}

For notational brevity, $\gamma$ in the arguments of \eqref{eq:apolcircle_center} and \eqref{eq:apolcircle_radius} will be omitted throughout the paper. The Apollonius circle provides a number of useful geometric properties, which are summarized below.

\begin{fact}\label{fact:fact1}
    For the game defined in Section~\ref{sec:problem_formulation}, optimal play implies that $\| \bm{C}(\mbx) - \x_E \| = \gamma R(\mbx)$ and $\| \bm{C}(\mbx) - \x_P \| = R(\mbx) / \gamma$, and thus $\| \bm{C}(\mbx) - \x_E \| = \gamma^2 \| \bm{C}(\mbx) - \x_P \|$. In the capture game, let $\x^\star \in \bR^2$ be the optimal capture point, then it holds under optimal play that $\| \x^\star - \x_{E,0} \| = \gamma \| \x^\star - \x_{P,0} \|$ and that $\x^\star$ lies on the Apollonius circle, i.e., $\| \x^\star - \bm{C}(\mbx) \| = R(\mbx)$. In the attack game, let $\x^\dagger \in \partial \Omega$ be the optimal attack point, then, again under optimal play, it follows that $\| \x^\dagger - \x_{E,0} \| = \gamma \| \x_P(t_a) - \x_{P,0} \|$.
\end{fact}

As mentioned, the Apollonius circle is the borderline that separates the dominant regions of the pursuer and the evader. The pursuer's dominant region, denoted by $\cD_P$, contains all the points that the pursuer can reach faster than the evader, whereas the evader's dominant region, $\cD_E$, includes all the points that the evader can reach at least as fast as the pursuer. Note that the circle $\cA$ itself is included in $\cD_E$ given the definitions of the terminal manifolds $\cT_c$ and $\cT_a$ given in Section~\ref{sec:problem_formulation}. Using the notion of these two dominant regions, the winning condition of each player is provided as follows.

\begin{proposition} \label{prop:dominantregions}
Consider the TPTDG defined in Section~\ref{sec:problem_formulation} and the dominant regions of the players, $\cD_P(\mbx,\gamma)$ and $\cD_E(\mbx,\gamma)$, whose expressions are
    \begin{align}
        &\cD_P(\mbx;\gamma) := \{ \z \in \bR^2 : \| \z - \bm{C}(\mbx) \| > R(\mbx) \}, \label{eq:dominantregion_pursuer}
        \\
        &\cD_E(\mbx;\gamma) := \{ \z \in \bR^2 : \| \z - \bm{C}(\mbx) \| \leq R(\mbx) \}. \label{eq:dominantregion_evader}
    \end{align}
    Then, under optimal play, the pursuer is guaranteed to win the game if $\cD_E \cap \Omega = \varnothing$, whereas the evader is guaranteed to win the game if $\cD_E \cap \Omega \neq \varnothing$.
\end{proposition}
\begin{proof}
    If $\cD_E \cap \Omega = \varnothing$, the pursuer is able to reach any point in $\Omega$ faster than the evader under optimal play, thus the first statement is proved. For the second statement, the fact that $\cD_E \cap \Omega \neq \varnothing$ implies that there exists a point in $\Omega$ that the evader can reach at least as fast as the pursuer under optimal play. It follows from the definition of $\cT_a$ that, if the evader reaches $\Omega$ faster than the pursuer, the evader wins. In the case when both players reach a point $\x \in \cD_E \cap \Omega$ at the same time (in view of Definition~\ref{def:apolloncirc}, the point $\x$ will belong to $\cA$), then by the definition of $\cT_c$, the joint state $\mbx$ will terminate at $\cT_a$, thus the evader wins.
\end{proof}

Now, given a point $\x \in \bR^2$, let us define the orthogonal projection of $\x$ onto $\Omega$ as $\proj_\Omega(\x) := \argmin\nolimits_{\z \in \Omega} \| \z - \x \|$. In TPTDGs, the existence and uniqueness of $\proj_\Omega (\x)$ are ensured for all $\x$ since $\Omega$ is a compact and convex set. In the following proposition, which is a generalized statement of Theorem 2 in \cite{garcia2019optimal}, the barrier surface that demarcates different winning sets is characterized.

\begin{proposition} \label{prop:barriersurface}
    Consider the TPTDG defined in Section~\ref{sec:problem_formulation}. The barrier surface, which divides the state space of the game $\bR^4$ into two winning sets $\cR_c$ (capture set) and $\cR_a$ (attack set), is given by $\cB = \{ \mbz \in \bR^4 : B(\mbz;\Omega,\gamma) = 0 \}$, where the barrier function $B:\bR^4 \rightarrow \bR$ is defined as
    \begin{equation} \label{eq:barrierfunction}
        B(\mbx;\Omega,\gamma) := \| \proj_\Omega \left( \bm{C}(\mbx) \right) - \bm{C}(\mbx) \| - R(\mbx).
    \end{equation}
    The winning sets $\cR_c$ and $\cR_a$ are expressed as
    \begin{align}
        &\cR_c := \{ \mbz \in \bR^4 : B(\mbz;\Omega,\gamma) > 0 \}, \label{eq:capturespace}
        \\
        &\cR_a := \{ \mbz \in \bR^4 : B(\mbz;\Omega,\gamma) \leq 0 \}. \label{eq:attackspace}
    \end{align}
    Then, if $\mbx \in \cR_c$, the pursuer is guaranteed to win the game, whereas if $\mbx \in \cR_a$, the evader is guaranteed to win the game, both under optimal play.
\end{proposition}
\begin{proof}
    Proposition~\ref{prop:dominantregions} implies $\proj_\Omega(\bm{C}(\mbx)) \in \cD_E$ if $\| \proj_\Omega (\bm{C}(\mbx)) - \bm{C}(\mbx) \| \leq R(\mbx)$ and $\proj_\Omega \in \cD_P$ if $\| \proj_\Omega (\bm{C}(\mbx)) - \bm{C}(\mbx) \| > R(\mbx)$. Since $\proj_\Omega (\bm{C}(\mbx)) \in \Omega$, then the inclusion $\proj_\Omega (\bm{C}(\mbx)) \in \cD_E$ implies $\cD_E \cap \Omega \neq \varnothing$, whereas the exclusion $\proj_\Omega (\bm{C}(\mbx)) \notin \cD_E$ implies $\cD_E \cap \Omega = \varnothing$. Hence, if $B(\mbx;\Omega,\gamma) > 0$, that is, $\| \proj_\Omega (\bm{C}(\mbx)) - \bm{C}(\mbx) \| > R(\mbx)$, it follows that $\proj_\Omega (\bm{C}(\mbx)) \in \cD_E$ and $\cD_E \cap \Omega \neq \varnothing$, and thus, in view of Proposition~\ref{prop:dominantregions}, the pursuer is guaranteed to win. Conversely, if $B(\mbx;\Omega,\gamma) \leq 0$, that is, $\| \proj_\Omega (\bm{C}(\mbx)) - \bm{C}(\mbx) \| \leq R(\mbx)$, it follows that $\proj_\Omega (\bm{C}(\mbx)) \in \cD_E$ and $\cD_E \cap \Omega = \varnothing$, and, consequently, the same proposition assures that the evader will win under optimal play.
\end{proof}

\begin{remark} \label{remark:semipermeable_barrier}
    The barrier surface $\cB$ is semipermeable, meaning that the state of the game $\mbx$ never crosses $\cB$ under optimal play, but if either of the players deviates from his/her own optimal path, $\mbx$ could cross $\cB$.
\end{remark}

If the pursuer's position $\x_P$ is known and fixed, the cross section of $\cB$, referred to as barrier curve and denoted by $\overline{\cB}$, can be visualized on the plane $\bR^2$. Simulation results that highlight Remark~\ref{remark:semipermeable_barrier} are illustrated in Section~\ref{sec:simulations}.

\section{The Games of Degree} \label{sec:games_of_degree}

Having found the solution of the game of kind, or the barrier function, we now solve two different games of degree: the capture game of degree and the attack game of degree. As discussed in Section~\ref{sec:problem_formulation}, the value functions of the TPTDG are defined differently in $\cR_c$ and $\cR_a$, thus two different sets of optimal state feedback strategies are to be found. Since the optimal strategies in open-loop form have been already found in Proposition~\ref{prop:hamiltonian}, similar to the solution procedure provided in \cite{garcia2019optimal}, we first formulate candidate state feedback strategies derived from those open-loop solutions, then we prove that these candidate solutions along with their corresponding value function satisfy the HJI equation.

\subsection{Capture Game of Degree}

In this subsection, the value function and the optimal strategies of the two players in the capture set $\cR_c$ are provided.

\begin{proposition} \label{prop:capgamedegree}
    For the TPTDG defined in Section~\ref{sec:problem_formulation}, the value function defined over $\cR_c$ is given by
    \begin{equation}
        V_{c}(\mbx) = \| \proj_\Omega (\bm{C}(\mbx)) - \bm{C}(\mbx) \| - R(\mbx),
    \end{equation}
    where $V_{c}(\mbx)$ is continuously differentiable for all $\mbx \in \cR_c$. The optimal state feedback strategies of the pursuer and the evader in $\cR_c$, denoted by $\uu_P^\star = (u_{P,x}^ \star,u_{P,y}^ \star) \in \cU$ and $\uu_E^\star = (u_{E,x}^ \star,u_{E,y}^ \star) \in \cU$, respectively, are defined as
    \begin{equation}
        \uu_P^\star =
        \dfrac{\x^\star - \x_P}{\| \x^\star - \x_P \|}, \quad \uu_E^\star =
        \dfrac{\x^\star - \x_E}{\| \x^\star - \x_E \|}, \label{eq:capturegame_optinputs}
    \end{equation}
    where the optimal capture point $\x^\star = (x^\star, y^\star) \in \bR^2$ is
    \begin{align}
        \x^\star &= \bigg( C_x + R(\mbx) \cdot \frac{P_x - C_x}{\sqrt{ (P_x - C_x)^2 + (P_y - C_y)^2 }}, \nonumber
        \\
        &\qquad C_y + R(\mbx) \cdot \frac{P_y - C_y}{ \sqrt{ (P_x - C_x)^2 + (P_y - C_y)^2 } } \bigg), \label{eq:capturegame_optcappoint}
    \end{align}
    with $\bm{C}(\mbx) = (C_x,C_y)$ and $\proj_\Omega (\bm{C}(\mbx)) = (P_x,P_y)$.
\end{proposition}
\begin{proof}
    We first show that the value function $V_{c}$ is continuously differentiable for all $\mbx \in \cR_c$. From \eqref{eq:valfun1} we obtain $V_{c} = \min_{\uu_E^\star(\cdot)}\max_{\uu_P^\star(\cdot)} J_c = \dist(\x^\star,\Omega)$. If $\mbx \in \cR_c$, it follows from Proposition~\ref{prop:dominantregions} and Fact~\ref{fact:fact1} that $\x^\star \in \cA$ and $\proj_\Omega (\bm{C}(\mbx)) \notin \cD_E$, thus $\| \x^\star - \bm{C}(\mbx) \| = R(\mbx)$. The value function then becomes
    \begin{align}
        V_{c}(\mbx) &= \| \proj_\Omega (\bm{C}(\mbx)) - \bm{C}(\mbx) \| - R(\mbx). \label{eq:capturegame_valuefunction_first}
    \end{align}
    Note that $\| \proj_\Omega (\bm{C}(\mbx)) - \bm{C}(\mbx) \| = \min_{\z \in \Omega} \| \z - \bm{C}(\mbx) \|$. To simplify differentiation, consider the following two angles: $\phi = \mathrm{atan2}(y_E - y_P,x_E - x_P) = \mathrm{atan2}(C_y - y_P,C_x - x_P)$ and $\theta = \mathrm{atan2}(P_y - C_y,P_x - C_x)$, where $\phi$ is the elevation angle of $\x_E$ (or equivalently $\bm{C}(\mbx)$) with respect to $\x_P$, and $\theta$ is the elevation angle of $\x^\star$ with respect to $\bm{C}(\mbx)$. Additionally, the partial derivatives of $\bm{C}(\mbx)$ and $R(\mbx)$ with respect to $\mbx$ are obtained as
    \begin{align}
        \frac{\partial \bm{C}(\mbx)}{\partial \mbx} &= \frac{1}{1-\gamma^2}
        \begin{bmatrix}
            -\gamma^2 & 0 & 1 & 0
            \\
            0 & -\gamma^2 & 0 & 1
        \end{bmatrix}, \label{eq:capturegame_Cpartial}
        \\
        \nabla_{\mbx} R(\mbx) &= \frac{\gamma}{1-\gamma^2} \big[ -\cos\phi ~ -\sin\phi ~ \cos\phi ~ \sin\phi \big]^\transpose. \label{eq:capturegame_Rpartial}
    \end{align}
    The gradient of $V_c$ in $\cR_c$ is then found as
    \begin{align}
        &\nabla_{\mbx} V_{c} = \nabla_{\mbx} \left( \| \proj_\Omega (\bm{C}(\mbx)) - \bm{C}(\mbx) \| - R(\mbx) \right) \nonumber
        \\
        &= \left[ \nabla_{\bm{C}}^\transpose (\| \proj_\Omega (\bm{C}(\mbx)) - \bm{C}(\mbx) \|) \frac{\partial \bm{C}(\mbx)}{\partial \mbx} \right]^\transpose - \nabla_\mbx R(\mbx) \nonumber
        \\
        &= - \left[ \frac{\partial \bm{C}(\mbx)}{\partial \mbx} \right]^\transpose \frac{\proj_\Omega (\bm{C}(\mbx)) - \bm{C}(\mbx)}{\| \proj_\Omega (\bm{C}(\mbx)) - \bm{C}(\mbx) \|} - \nabla_\mbx R(\mbx). \label{eq:capturegame_gradient_before}
    \end{align}
    Substituting $(\proj_\Omega (\bm{C}(\mbx)) - \bm{C}(\mbx))/\| \proj_\Omega (\bm{C}(\mbx)) - \bm{C}(\mbx) \| = (\cos\theta,\sin\theta)$ into \eqref{eq:capturegame_gradient_before} along with the partial derivatives of $\bm{C}(\mbx)$ in \eqref{eq:capturegame_Cpartial} and $R(\mbx)$ in \eqref{eq:capturegame_Rpartial} yields
    \begin{align}
        \nabla_\mbx V_{c} = \frac{1}{1-\gamma^2}
        \begin{bmatrix}
            \gamma^2 \cos\theta + \gamma \cos\phi
            \\
            \gamma^2 \sin\theta + \gamma \sin\phi
            \\
            -\cos\theta - \gamma \cos\phi
            \\
            -\sin\theta - \gamma \sin\phi
        \end{bmatrix}.
    \end{align}
    Since $0 < \gamma < 1$ and no $\mbx$ such that $\| \x_P - \x_E\| = 0$ belongs to $\cR_c$, $\phi$ and $\theta$ are always determinate, and thus $V_{c}$ is continuously differentiable for all $\mbx \in \cR_c$.
    
    Next, the optimal capture point $\x^\star$ are rewritten as $\x^\star = ( C_x + R(\mbx) \cos\theta, C_y + R(\mbx) \sin\theta )$. It follows from Fact~\ref{fact:fact1} that $\| \x^\star - \x_E \| = \gamma \| \x^\star - \x_P \|$ and $\| \bm{C}(\mbx) - \x_E \| = \gamma^2 \| \bm{C}(\mbx) - \x_P \| = \gamma R(\mbx)$. Using these properties, the optimal strategies in \eqref{eq:capturegame_optinputs} are rewritten as
    \begin{align}
        \uu_P^\star &= \left( R(\mbx) \frac{ \cos\phi + \gamma\cos\theta }{\gamma \| \x^\star - \x_P \|}, R(\mbx) \frac{ \sin\phi + \gamma\sin\theta}{\gamma \| \x^\star - \x_P \|} \right),
    \end{align}
    \begin{align}
        \uu_E^\star &= \left( R(\mbx) \frac{ \gamma\cos\phi + \cos\theta }{ \| \x^\star - \x_P \|}, R(\mbx) \frac{ \gamma\sin\phi + \sin\theta }{ \| \x^\star - \x_P \|} \right).
    \end{align}
    Now, the HJI equation adjusted for our problem is given by
    \begin{equation}\label{eq:HJIpf}
        -\frac{\partial V_{c}}{\partial t} = \nabla_\mbx V_{c} \cdot \bm{f}(\mbx,\uu_P^\star,\uu_E^\star).
    \end{equation}
    The value function $V_{c}$ does not depend on time explicitly and thus $\partial V_{c}/\partial t = 0$. Thus, \eqref{eq:HJIpf} takes the following form:
    \begin{align}
        &\nabla_\mbx V_{c} \cdot  \bm{f}(\mbx,\uu_P^\star,\uu_E^\star) = \nabla_\mbx V_{c} \cdot v_P (\uu_P^{\star}, \gamma \uu_E^{\star}) \nonumber
        \\
        &\qquad = \frac{v_P R(\mbx)}{1-\gamma^2}
        \begin{bmatrix}
            \gamma^2 \cos\theta + \gamma \cos\phi
            \\
            \gamma^2 \sin\theta + \gamma \sin\phi
            \\
            -\cos\theta - \gamma \cos\phi
            \\
            -\sin\theta - \gamma \sin\phi
        \end{bmatrix}
        \cdot
        \begin{bmatrix}
            \frac{  \cos\phi + \gamma\cos\theta}{\gamma\| \x^\star - \x_P \|}
            \\
            \frac{  \sin\phi + \gamma\sin\theta }{\gamma\| \x^\star - \x_P \|}
            \\
            \frac{  \gamma\cos\phi + \cos\theta }{\| \x^\star - \x_P \|}
            \\
            \frac{  \gamma\sin\phi + \sin\theta }{\| \x^\star - \x_P \|}
        \end{bmatrix} \nonumber
        \\
        &\qquad = 0.
    \end{align}
    This completes the proof.
\end{proof}

\begin{remark}
    Since $V_{c}(\mbx)$ is continuously differentiable for all $\mbx \in \cR_c$, there exists no dispersal surface in the capture set $\cR_c$. In other words, the optimal strategies for the pursuer and evader in $\cR_c$ are always unique.
    \label{remark:capgame_semiperm}
\end{remark}

\subsection{Attack Game of Degree}

In this subsection, a proposition similar to Proposition~\ref{prop:capgamedegree} is provided for the attack game of degree.

\begin{proposition} \label{prop:attgame_optinputs}
    For the TPTDG defined in Section~\ref{sec:problem_formulation}, the value function defined over $\cR_a$ is given by
    \begin{equation} \label{eq:prop6valuefunc}
        V_a(\mbx) = \| \x^\dagger - \x_P \| - \frac{1}{\gamma} \| \x^\dagger - \x_E \|,
    \end{equation}
    where $\x^\dagger = (x^\dagger,y^\dagger) \in \bR^2$ is the optimal attack point which corresponds to the unique global minimizer of the following constrained convex optimization problem:
    \begin{equation}\label{eq:attackgame_optattpoint}
    \begin{aligned}
       \textrm{minimize}_{\z \in \bR^2} &\quad -\| \z - \x_P \| + \frac{1}{\gamma} \| \z - \x_E \|,
        \\
        \textrm{subject to} &\quad h(\z) = 0,
        \\
        \textrm{and} &\quad \| \z - \x_E \| \leq \gamma \| \z - \x_P \|.
    \end{aligned}
   \end{equation}
   Then, $V_a(\mbx)$ is continuously differentiable for all $\mbx \in \cR_a$, and the optimal state feedback strategies of the pursuer and the evader in $\cR_a$ are respectively defined as
    \begin{equation}
        \uu_P^\star =
        \dfrac{\x^\dagger - \x_P}{\| \x^\dagger - \x_P \|}, \quad \uu_E^\star =
        \dfrac{\x^\dagger - \x_E}{\| \x^\dagger - \x_E \|},
        \label{eq:attackgame_optinputs}
    \end{equation}
    where $\uu_P^\star = (u_{P,x}^ \star,u_{P,y}^ \star) \in \cU$ and $\uu_E^\star = (u_{E,x}^ \star,u_{E,y}^ \star) \in \cU$.
\end{proposition}

\begin{proof}
    From \eqref{eq:valfun2}, $V_a(\mbx) = \min_{\uu_E^\star(\cdot)}\max_{\uu_P^\star(\cdot)} J_e = \dist(\x^\dagger,\x_P(t_a))$. Since the optimal trajectories of the pursuer and evader are straight lines, the distance between the two players at the final time is $\|\x^\dagger - \x_P\| - \frac{1}{\gamma} \| \x^\dagger - \x_E \|$, which leads to the expression in \eqref{eq:prop6valuefunc}. Next, to simplify the differentiation, let us define the following two angles: $\psi = \mathrm{atan2}((y^\dagger - y_P), (x^\dagger - x_P))$ and $\varphi = \mathrm{atan2}((y^\dagger - y_E), (x^\dagger - x_E))$, where $\psi$ (resp., $\varphi$) is the elevation angle of $\x^\dagger$ with respect to $\x_P$ (resp., $\x_E$). The gradient of $V_a$ in $\cR_a$ is then given by
    \begin{align}
        \nabla_{\mbx}& V_a = \nabla_{\mbx} \left( \| \x^\dagger - \x_P \| - \frac{1}{\gamma}\| \x^\dagger - \x_E \| \right) \nonumber
        \\
        &=
        -\begin{bmatrix}
            \dfrac{x^\dagger-x_P}{\| \x^\dagger - \x_P \|}
            \\
            \dfrac{y^\dagger-y_P}{\| \x^\dagger - \x_P \|}
            \\
            0
            \\
            0
        \end{bmatrix}
        + \frac{1}{\gamma}
        \begin{bmatrix}
            0
            \\
            0
            \\
            \dfrac{x^\dagger-x_E}{\| \x^\dagger - \x_E \|}
            \\
            \dfrac{y^\dagger-y_E}{\| \x^\dagger - \x_E \|}
        \end{bmatrix} \nonumber
        \\
        &=
        \begin{bmatrix}
            -\cos\psi & -\sin\psi & (\cos\varphi)/\gamma & (\sin\varphi)/\gamma
        \end{bmatrix}^\transpose,
    \end{align}
    where $\psi$ and $\varphi$ are always determinate since $\| \x_E - \x_P \| \neq 0$ for all $t \in [0,t_a]$ (otherwise the game would terminate). The optimal strategies can be rewritten using $\psi$ and $\varphi$ as $\uu_P^\star = (\cos\psi,\sin\psi), \quad \uu_E^\star = (\cos\varphi,\sin\varphi)$. Then, again, $\partial V_a / \partial t = 0$, and the HJI equation is solved as
    \begin{align}
        \nabla_\mbx V_a \cdot  \bm{f}&(\mbx,\uu_P^\star,\uu_E^\star) = \nabla_\mbx V_a \cdot v_P(\uu_P^{\star}, \gamma \uu_E^{\star}) \nonumber
        \\
        &= \frac{v_P}{\gamma}
        \begin{bmatrix}
            -\gamma \cos\psi
            \\
            -\gamma \sin\psi
            \\
            \cos\varphi
            \\
            \sin\varphi
        \end{bmatrix}
        \cdot
        \begin{bmatrix}
            (x^\star - x_P)/\| \x^\star - \x_P \|
            \\
            (y^\star - y_P)/\| \x^\star - \x_P \|
            \\
            \gamma (x^\star - x_E)/\| \x^\star - \x_E \|
            \\
            \gamma (y^\star - y_E)/\| \x^\star - \x_E \|
        \end{bmatrix} \nonumber
        \\
        &= \frac{v_P}{\gamma}
        \begin{bmatrix}
            -\gamma\cos\psi
            \\
            -\gamma\sin\psi
            \\
            \cos\varphi
            \\
            \sin\varphi
        \end{bmatrix}
        \cdot
        \begin{bmatrix}
            \cos\psi
            \\
            \sin\psi
            \\
            \gamma\cos\varphi
            \\
            \gamma\sin\varphi
        \end{bmatrix} = 0.
    \end{align}
    This completes the proof.
\end{proof}

\begin{remark}
    Similar to Remark~\ref{remark:capgame_semiperm}, since $V_a$ is continuously differentiable for all $\mbx \in \cR_a$, there exists no dispersal surface in $\cR_a$, and thus the optimal strategies of the pursuer and evader are always unique therein.
    \label{remark:attgame_semiperm}
\end{remark}

\section{Numerical Simulations} \label{sec:simulations}

In this section, we present numerical simulations that illustrate the main results of this paper. Two different scenarios are considered to visualize the subgames discussed in Section~\ref{sec:games_of_degree}. In the first scenario, both the pursuer and the evader play their optimal strategies provided in Section~\ref{sec:games_of_degree}, whereas in the second scenario one of the players plays a nonoptimal strategy. The common parameters are selected as $\x_{P,0} = (1,-0.5),~\x_E = (3,1),~v_P = 1$, and $v_E = 0.7$ such that $\gamma = 0.7$. The implicit function for the target boundary $\partial \Omega$, which is illustrated as a black curve in Figures~\ref{fig:capgame} and \ref{fig:attgame}, is chosen to be $h(x,y) = x^4 + x^2y^2 - 2xy^2 - 2y^3 + y^4$. The barrier curve at $t_0$ (or $t=0$), $\overline{\cB}_0$, is numerically computed and illustrated in the same figures as a red curve.

From Proposition~\ref{prop:barriersurface}, we know that, if the joint state $\mbx \in \cR_c$ (resp, $\mbx \in \cR_a$), the pursuer (resp., the evader) is guaranteed to win the game under optimal play, which is illustrated in Figure~\ref{fig:capgame_opt} (resp., Figure~\ref{fig:attgame_opt}). In Figure~\ref{fig:capgame_opt}, the pursuer and the evader employ their optimal strategy given in \eqref{eq:capturegame_optinputs}, which results in linear trajectories for both players. In the same figure, the optimal capture point at $t_0$, $\x^\star(t_0)$, is given by \eqref{eq:capturegame_optcappoint} and is indicated by the blue square, whereas the projection of the center of the Apollonius circle between $\x_{P,0}$ and $\x_{E,0}$ onto $\partial \Omega$, $\proj_\Omega(\bm{C}(\mbx_0))$, is indicated by the orange square. In Figure~\ref{fig:attgame_opt}, the players employ their optimal strategies given in \eqref{eq:attackgame_optinputs}, again resulting in linear trajectories for both players. The optimal attack point at $t_0$, where the evader actually arrives at the final time $t_f$, is found by solving the constrained optimization problem \eqref{eq:attackgame_optattpoint} and is indicated by the red square. The pursuer ends up at $\x_P(t_f)$ which is the closest point to $\x^\dagger(t_0)$ that he can reach.

\begin{figure}
    \centering
    \subfigure[Optimal play]{\includegraphics[width=4.25cm]{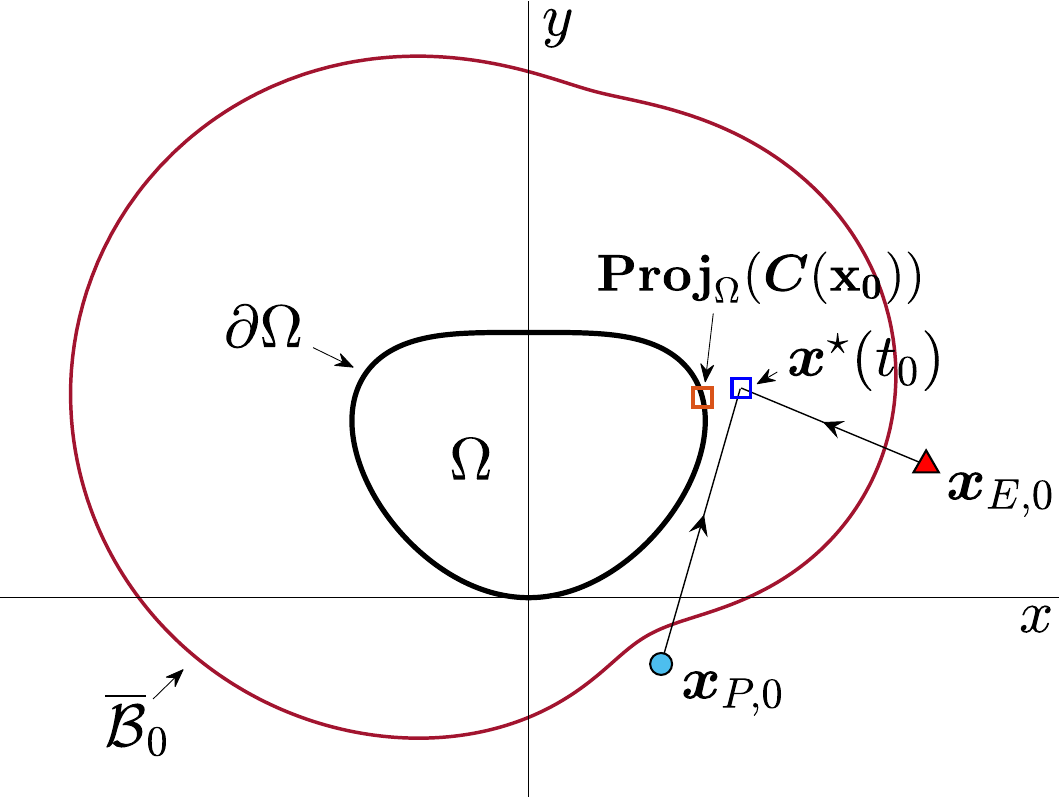}\label{fig:capgame_opt}}
    \subfigure[Nonoptimal play by pursuer]{\includegraphics[width=4.25cm]{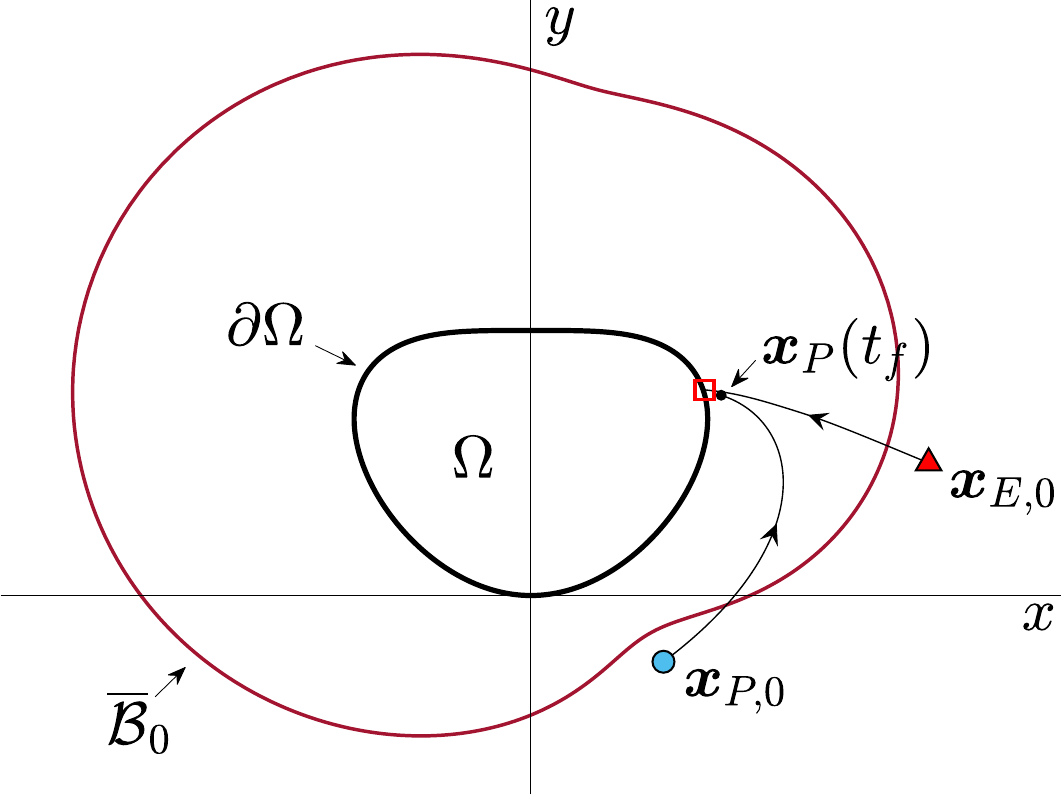}\label{fig:capgame_nopt}}
    \caption{Two different outcomes of TPTDGs with $\mbx_0 \in \cR_c$}
    \label{fig:capgame}
\end{figure}

As addressed in Remark~\ref{remark:semipermeable_barrier}, the barrier surface $\cB$ (or the barrier curve $\overline{\cB}$) is semipermeable and may be crossed under nonoptimal play. This is illustrated in Figures~\ref{fig:capgame_nopt} and \ref{fig:attgame_nopt}. In Figure~\ref{fig:capgame_nopt}, the evader employs her optimal evasion strategy given in \eqref{eq:capturegame_optinputs}, whereas the pursuer employs a pure pursuit strategy. The joint state $\mbx$ crosses $\cB$ at some time $t_s \in [t_0,t_f]$ (i.e., $B(\mbx(t_s))=0$), and the game transitions from the capture game to the attack game. Consequently, the evader successfully reaches $\partial \Omega$ before captured by the pursuer. Similarly, in Figure~\ref{fig:attgame_nopt}, the pursuer employs his optimal strategy in \eqref{eq:attackgame_optinputs}, whereas the evader moves along a (nonoptimal) linear path toward an arbitrary point on $\partial \Omega$. The game then retreats from the attack game to the capture game at some time $t_s \in [t_0,t_f]$ (again, $B(\mbx(t_s))=0$), and the evader is eventually captured by the pursuer before reaching $\partial \Omega$. Discussion on such $t_s$ are omitted for brevity.

\section{Conclusions} \label{sec:concl}

In this paper, a two-player game of guarding a compact convex target set with smooth boundary by a single pursuer is addressed based on the combination of geometric and differential game theoretic methods. The main contributions of this work include the characterization of the generalized barrier function (for a game of kind) as well as generalized value functions and optimal state feedback strategies for both players (for games of degree). The proposed strategies are respectively shown to be the saddle point of the capture game and the attack game via the Hamilton-Jacobi-Isaacs equation. In our future work, we will leverage the solution approach proposed herein to study the characteristics of barrier surfaces and optimal strategies in high-dimensional multiplayer target-defense games.

\begin{figure}
    \centering
    \subfigure[Optimal play]{\includegraphics[width=4.25cm]{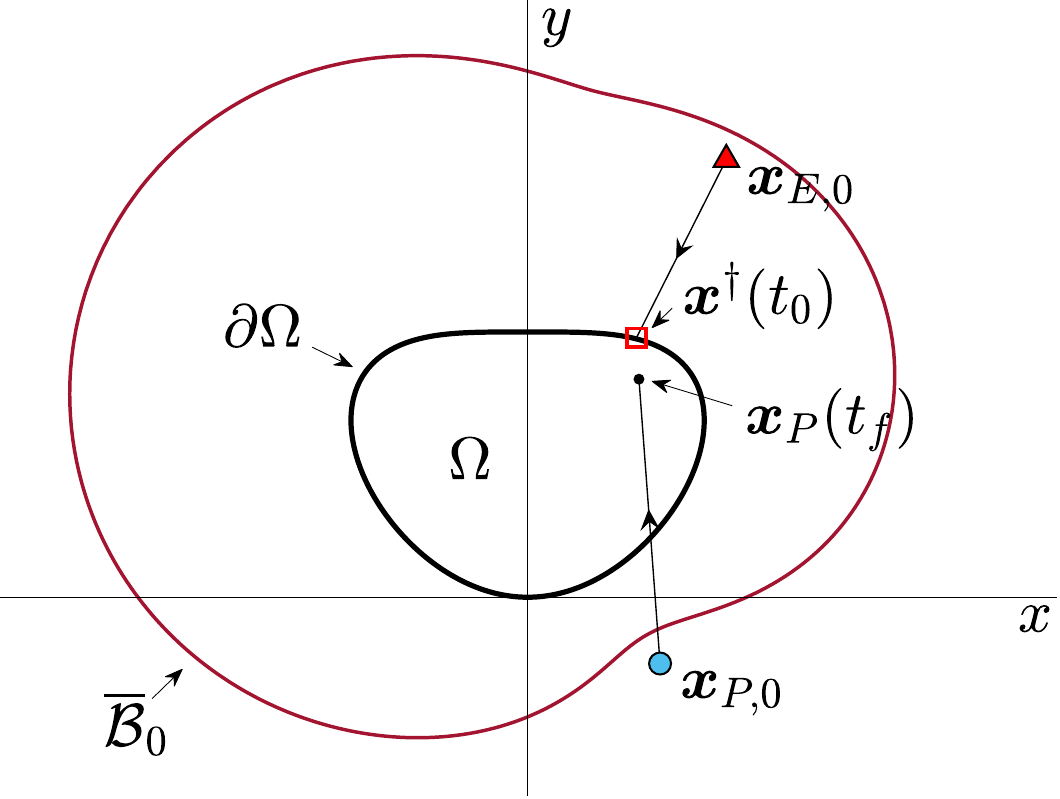}\label{fig:attgame_opt}}
    \subfigure[Nonoptimal play by evader]{\includegraphics[width=4.25cm]{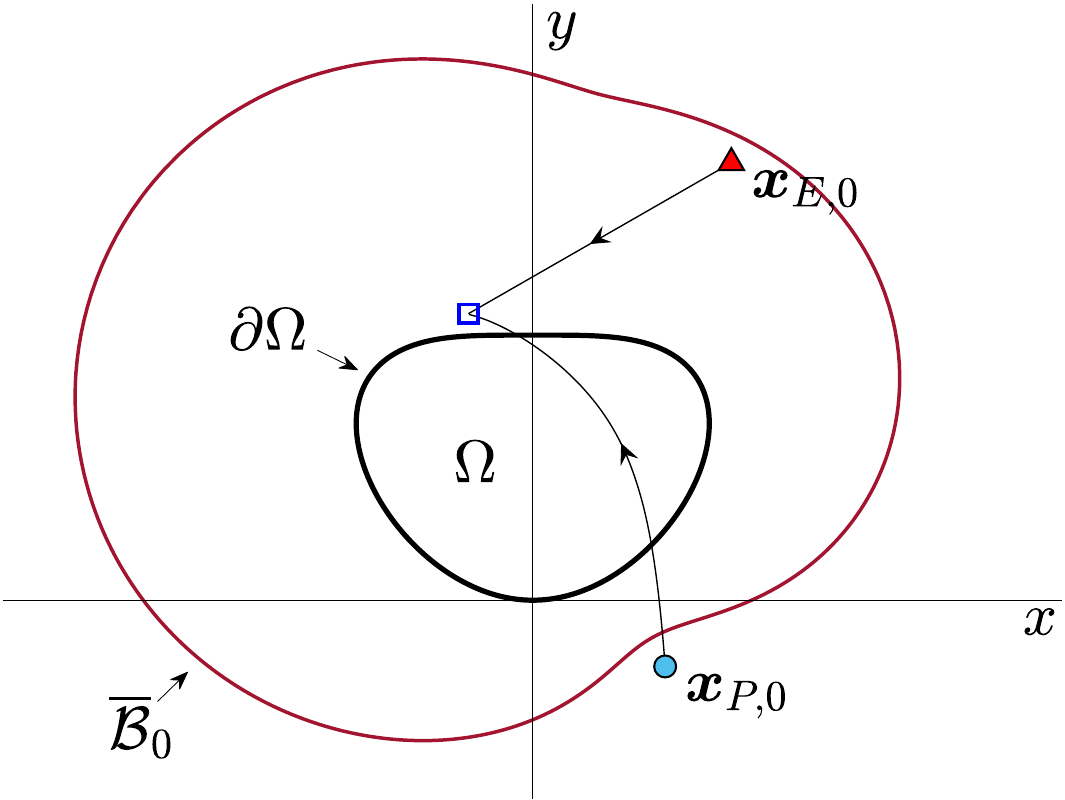}\label{fig:attgame_nopt}}
    \caption{Two different outcomes of TPTDGs with $\mbx_0 \in \cR_a$}
    \label{fig:attgame}
\end{figure}

\bibliographystyle{ieeetr}
\bibliography{pegref}

\end{document}